\documentclass[10pt,A4paper,conference]{IEEEtran}
\usepackage{amsthm}
\usepackage{amsmath}
\usepackage{amsfonts}
\usepackage{graphicx}
\usepackage{latexsym}
\usepackage{amssymb}
\usepackage{stmaryrd}
\usepackage{amscd}
\usepackage[small]{caption}
\usepackage{array}
\usepackage{tikz,ifthen,calc}
\usetikzlibrary{trees,snakes,automata}
\usetikzlibrary{arrows,shapes,positioning}
\usepackage{paralist}
\usepackage{slashbox}

\newcommand{\deff}{\mbox{$\stackrel{\rm def}{=}$}}

\newcommand{\sbinom}[2]{\left[ \begin{array}{c} #1 \\ #2 \end{array} \right] }

\newcommand{\field}[1]{\mathbb{#1}}

\newcommand{\Z}{\field{Z}}

\newcommand{\cA}{{\cal A}}

\newcommand{\cS}{{\cal S}}

\newcommand{\sP}{\field{P}}

\newcommand{\sG}{\field{G}}

\DeclareMathAlphabet{\mathbfsl}{OT1}{cmr}{bx}{it}
\newcommand{\uuu}{\kern-1pt\mathbfsl{u}\kern-0.5pt}
\newcommand{\vvv}{\kern-1pt\mathbfsl{v}\kern-0.5pt}

\newcommand{\myboxplus}{\kern1pt\mbox{\small$\boxplus$}}

\makeatletter \DeclareRobustCommand{\sbinom}{\genfrac[]\z@{}}
\makeatother
\newcommand{\G}[2]{\sbinom{{#1}\kern-1pt}{{#2}\kern-1pt}}
\newcommand{\Gq}[2]{\sbinom{{#1}\kern-0.25pt}{{#2}\kern-0.25pt}}

\newcommand{\Ps}{\smash{{\sP\kern-2.0pt}_q\kern-0.5pt(n)}}
\newcommand{\sPs}{\smash{{\sP\kern-1.5pt}_q(n)}}
\newcommand{\Ptwo}{\smash{{\sP\kern-2.0pt}_2\kern-0.5pt(n)}}
\newcommand{\Ptwom}{\smash{{\sP\kern-2.0pt}_2\kern-0.5pt(m)}}
\newcommand{\Ptwonm}{\smash{{\sP\kern-2.0pt}_2\kern-0.5pt(n+m)}}
\newcommand{\Ptwoa}{\smash{{\sP\kern-2.0pt}_2\kern-0.5pt(1)}}
\newcommand{\Ptwob}{\smash{{\sP\kern-2.0pt}_2\kern-0.5pt(2)}}
\newcommand{\Ptwoc}{\smash{{\sP\kern-2.0pt}_2\kern-0.5pt(3)}}
\newcommand{\Ptwod}{\smash{{\sP\kern-2.0pt}_2\kern-0.5pt(4)}}
\newcommand{\Ptwoe}{\smash{{\sP\kern-2.0pt}_2\kern-0.5pt(5)}}
\newcommand{\Ptwof}{\smash{{\sP\kern-2.0pt}_2\kern-0.5pt(6)}}
\newcommand{\Ptwokm}{\smash{{\sP\kern-2.0pt}_2\kern-0.5pt(2k-1)}}
\newcommand{\Pone}{\smash{{\sP\kern-2.5pt}_2\kern-0.5pt(n{-}1)}}

\newcommand{\Gr}{\smash{{\sG\kern-1.5pt}_q\kern-0.5pt(n,k)}}
\newcommand{\Gi}{\smash{{\sG\kern-1.5pt}_q\kern-0.5pt(n,i)}}
\newcommand{\Gj}{\smash{{\sG\kern-1.5pt}_q\kern-0.5pt(n,j)}}
\newcommand{\Grmk}{\smash{{\sG\kern-1.5pt}_q\kern-0.5pt(n,n-k)}}
\newcommand{\Grdk}{\smash{{\sG\kern-1.5pt}_q\kern-0.5pt(2k,k)}}
\newcommand{\Grekappa}{\smash{{\sG\kern-1.5pt}_q\kern-0.5pt(n,e+1-\kappa)}}
\newcommand{\Grtwoekappa}{\smash{{\sG\kern-1.5pt}_q\kern-0.5pt(n,2e+1-\kappa)}}
\newcommand{\Gremkappa}{\smash{{\sG\kern-1.5pt}_q\kern-0.5pt(n,e-\kappa)}}
\newcommand{\Gn}{\smash{{\sG\kern-1.5pt}_2\kern-0.5pt(n,n{-}1)}}
\newcommand{\Gnq}{\smash{{\sG\kern-1.5pt}_q\kern-0.5pt(n,n{-}1)}}
\newcommand{\Gone}{\smash{{\sG\kern-1.5pt}_2\kern-0.5pt(n,1)}}
\newcommand{\Gqone}{\smash{{\sG\kern-1.5pt}_q\kern-0.5pt(n,1)}}
\newcommand{\GTwo}{\smash{{\sG\kern-1.5pt}_2\kern-0.5pt(n,k)}}
\newcommand{\GTwonk}[2]{{\smash{{\sG\kern-1.5pt}_2\kern-0.5pt({#1},{#2})}}}
\newcommand{\Gnk}{\smash{{\sG\kern-1.5pt}_2\kern-0.5pt(n,n{-}k)}}
\newcommand{\Greone}{\smash{{\sG\kern-1.5pt}_q\kern-0.5pt(n,e{+}1)}}
\newcommand{\Gretwo}{\smash{{\sG\kern-1.5pt}_q\kern-0.5pt(n,e{+}2)}}

\newcommand{\be}[1]{\begin{equation}\label{#1}}
\newcommand{\ee}{\end{equation}}

\newcommand{\Cref}[1]{Co\-rol\-la\-ry\,\ref{#1}}

\newtheorem{theorem}{Theorem}

\begin{document}

\title{Local Rank Modulation for Flash Memories \vspace{-1.0ex}}

\author{\authorblockN{Michal Horovitz}
\authorblockA{Dept. of Computer Science\\
Technion-Israel Institute of Technology\\
Haifa 32000, Israel \\
Email: michalho@cs.technion.ac.il} \and
\authorblockN{Tuvi Etzion}
\authorblockA{Dept. of Computer Science\\
Technion-Israel Institute of Technology\\
Haifa 32000, Israel \\
Email: etzion@cs.technion.ac.il}}

\maketitle
\begin{abstract}
Local rank modulation scheme was suggested recently for
representing information in flash memories
in order to overcome drawbacks of rank modulation.
For $0 < s\leq t\leq n$ with $s$ divides $n$, an $(s,t,n)$-LRM scheme
is a local rank modulation scheme
where the $n$ cells are locally viewed cyclically through a
sliding window of size $t$
resulting in a sequence of small permutations which
requires less comparisons and less distinct values.
The gap between two such windows equals to~$s$.
In this work, encoding, decoding, and asymptotic enumeration of the
$(1,3,n)$-LRM scheme is studied.
The techniques which are suggested have some generalizations
for $(1,t,n)$-LRM, $t > 3$, but the proofs will become more
complicated. The enumeration problem is presented
also as a purely combinatorial problem.
Finally, we prove the conjecture that the size
of a constant weight $(1,2,n)$-LRM Gray code with weight two is at most $2n$.
\end{abstract}

\section{Introduction}
Flash memory is a non-volatile technology that is both electrically
programmable and electrically erasable. It incorporates a set of cells
maintained at a set of levels of charge to encode information.
While raising the charge level of a cell is an easy operation,
reducing the charge level requires the erasure of the whole block to which the cell belongs.
For this reason charge is injected into the cell over several iterations.
Such programming is slow and can cause errors since cells may be injected with extra unwanted charge.
Other common errors in flash memory cells are due to charge leakage and reading disturbance that
may cause charge to move from one cell to its adjacent cells.
In order to overcome these problems, the novel
framework of \emph{rank modulation} was introduced in~\cite{rankModulation1_2009}.
In this setup, the information is carried by the relative ranking of the
cells' charge levels and not by the absolute values of the charge levels.
Denote the charge level in the $i$th cell by $c_i$, $0\leq i < n$,
and $c=(c_0,c_2,\ldots, c_{n-1})$ is the sequence of the charges in~$n$ cells.
A codeword in this scheme is the permutation defined by the
order of the charge levels, from the highest one to the lowest one, e.g.
if $n=5$ and $c=(3,5,2,7,10)$ then the permutation,
i.e., the codeword in the rank modulation scheme, is $[5,4,2,1,3]$.
This allows for more efficient programming of cells, and coding by the ranking of the cells' charge levels
is more robust to charge leakage than coding by their actual values.
The \emph{push-to-the-top} operation is a basic minimal cost
operation in the rank modulation scheme by which a single cell has its charge
level increased such that it will be the highest of the set.

A drawback of the rank modulation scheme is the need for a large number of
comparisons when reading the induced permutation. Furthermore, distinct
$n$ charge levels are required for a group of~$n$ cells.
The \emph{local rank modulation} scheme was suggested in order to overcome these problems.
In this scheme, the $n$~cells are locally viewed through a
sliding window, resulting in a sequence of permutations
for a much smaller number of cells which
requires less comparisons and less distinct values.
For $0 < s\leq t\leq n$, where $s$ divides $n$,
the $(s,t,n)$-LRM scheme, defined in \cite{GLSB11,WJB09},
is a local rank modulation scheme over~$n$ physical cells,
where~$t$ is the size of each sliding window and~$s$ is the gap between two such windows.
In this scheme the permutations are over $\{1,2,\ldots, t\}$,
i.e., form $S_t$,
and the push-to-the-top operation
merely raises the charge level of the selected cell
above those cells which are comparable with it.
We say a sequence with $\frac{n}{s}$ permutations from $S_t$ is an
$(s,t,n)$-LRM scheme \emph{realizable} if it can be demodulated
to a sequence of charges in~$n$ cells under the $(s,t,n)$-LRM scheme.
Except for the degenerate case where $s=t=n$, not every sequence is realizable.

The $(1,2,n)$-LRM scheme was defined in \cite{GLSB11}
in order to get the simplest hardware implementation.
The demodulated sequences of permutations
in this scheme contain all the binary words except two, the all-ones and all-zeros sequences.
Therefore, the number of codewords in this scheme is $2^n-2$.

In this paper we focus on the $(1,t,n)$-LRM schemes for $t\geq 3$,
and suggest a demodulation method for these schemes. The
$(1,t,n)$-LRM scheme is a local rank modulation scheme over $n$ physical cells,
where the size of each sliding window is $t$, and each cell starts a new window.
Since the size of a sliding window is $t$,
demodulated sequences of permutations in this scheme contain $t!$ permutations.
Therefore, we need $t!$ symbols to present the demodulated sequences of permutations.

Let $s=(s_1,s_2,\ldots,s_{t!})$ be an order of the $t!$ permutations from $S_t$,
and $\Sigma =\{1,2, \ldots, t!\}$ be an alphabet where $i$ represents the permutation $s_i$.
A sequence $\alpha=(\alpha_0,\alpha_1,\ldots, \alpha_{n-1})$ over the alphabet $\Sigma$
is called a \emph{base-word} in the $(1,t,n)$-LRM scheme, and it
is realizable, if there exists a sequence of charges
$c=(c_0,c_1,\ldots , c_{n-1})$, such that for each $i$, $0\leq i \leq n-1$,
$\alpha_i$ represent the permutation induced by $c_i, c_{i+1}, \ldots, c_{i+t-1}$,
where indices are taken modulo $n$. The indices in the base-words and codewords
are also taken modulo $n$ as in the charge levels.

In this paper we produce a mapping method, in which each $\alpha$,
a base-word over the alphabet of size $t!$,
is mapped to a codeword $g=(g_0, g_1,\ldots, g_{n-1})$ over an alphabet of size $t$.
A codeword is called \emph{legal}
if there exists a realizable base-word which is mapped to it.
We have to make sure that
two distinct realizable base-words
are mapped into two distinct legal codewords.

Let $M_t$ be the number of legal codewords in the $(1,t,n)$-LRM scheme.
Clearly, $M_t \leq t^n$, but
this upper bound is not tight since there exist illegal codewords.
We conjecture that $\lim\limits_{n\to \infty}\frac{M_t}{t^n}=1$ and
prove this conjecture for $t=3$ and $t=4$.

The rest of this paper is organized as follows.
The encoding, decoding and asymptotic enumeration of $(1,3,n)$-LRM
scheme is presented in Section \ref{sec:13nScheme}.
Generalizations, especially for the enumeration technique
for the $(1,t,n)$-LRM scheme, $t>3$, is
given in Section~\ref{sec:1tnScheme}. The generalization
of the asymptotic enumeration problem is presented as
a combinatorial problem. The solution for the $(1,4,n)$-LRM scheme is also given.
In Section~\ref{sec:weight} it is proved that
the size of a constant weight $(1,2,n)$-LRM Gray code
with weight two is at most $2n$. Thus, proving a conjecture
from~\cite{GLSB11}.
In Section~\ref{sec:conclusion} conclusion and problems for future research
are presented.

\section{The $(1,3,n)$-LRM scheme}
\label{sec:13nScheme}

In the $(1,3,n)$-LRM scheme the size of each sliding window is~$3$.
Therefore, an alphabet of size $3!$ is required
to present the demodulated sequences of permutations.
\begin{center}
\begin{tabular}{rrr}
$s_1=[1,2,3]$ & $s_2=[1,3,2]$ \\
$s_3=[2,1,3]$ & $s_4=[3,1,2]$ \\
$s_5=[2,3,1]$ & $s_6=[3,2,1]$ \\
\end{tabular}
\end{center}

The alphabet of the base-words is $\Sigma=\{1,2,\ldots, 6\}$,
where the symbol $\ell$ represents the permutation $s_{\ell}$.
Let ${\alpha=(\alpha_0,\alpha_1,\ldots, \alpha_{n-1})}$ be a base-word.
Note that the last two cells which determine $\alpha_i$ ($0\leq i\leq n-1$)
are the first two cells which determine $\alpha_{i+1}$,
i.e., the permutation related to $\alpha_{i+1}$ is
obtained from $\alpha_i$ by the following way.
The symbol~$1$ in the permutation related to $\alpha_i$ is omitted,
the symbols $2,3$ in the permutation are replaced with $1,2$, respectively,
and a new symbol~$3$ is
inserted before~$1,2$, between them, or after both of them.
Therefore, given $\alpha_i$, there are exactly $3$ options for $\alpha_{i+1}$.

Let $\Sigma^1=\{1,3,5\}$ and $\Sigma^2=\{2,4,6\}$
be a partition of $\Sigma$ into the even and
the odd symbols, respectively.
Note that for each $\Sigma^i$, the permutations related to the
symbols in $\Sigma^i$ agree on the order of cells 2 and 3.
Therefore, they also agree on the $3$ possibilities of their succeeding permutation.
Denote the set of symbols of
these succeeding permutations by $\tilde{\Sigma}^{i}$.
Thus, we have $\tilde{\Sigma}^1=\{1,2,4\}$ and $\tilde{\Sigma}^2=\{3,5,6\}$.

The base-word $\alpha$ is mapped to a codeword
$g=(g_0,g_1,\ldots,g_{n-1})$ over the alphabet $\{0,1,2\}$.
The relations between $\alpha_{i-1}$, $\alpha_i$,
and $g_i$, where $0\leq i\leq n-1$,
are presented in Table~\ref{table:encodingKey13n}.
This table induces a mapping from the realizable base-words to the codewords.
As mentioned before, given $\alpha_{i-1}$, there are three options for $\alpha_i$.
In all these options the sub-permutation of $\{1,2\}$ is the same,
and the difference is the index of symbol~$3$ in the
permutation related to $\alpha_i$.
Thus, $g_i$ represents the index of symbol $3$ in this permutation
and it equal to the number of symbols which are
to the right of the symbol~$3$ in the
permutation related to $\alpha_i$.
In other words, $g_i$ represents the relation
between $c_{i+2}$, the charge level in cell $i+2$,
and the charge levels in two cells which proceed
it, i.e., $c_i$ and $c_{i+1}$.
\begin{table}[!ht]
\begin{center}
\begin{tabular}{|l||l|l|l|}
\hline
$\alpha_{i-1}\in S^1$ & $\alpha_i=1$ & $\alpha_i=2$ & $\alpha_i=4$ \\
\hline
$\alpha_{i-1}\in S^2$ & $\alpha_i=3$ & $\alpha_i=5$ & $\alpha_i=6$ \\
\hline
\hline
 & $g_i=0$ & $g_i=1$ & $g_i=2$ \\
\hline
\end{tabular}
\end{center}
\caption{The encoding key of the $(1,3,n)$-LRM scheme}
\label{table:encodingKey13n}
\end{table}

\vspace{-0.3cm}

Note that there might exist non-realizable base-words
which are mapped to codewords by this method.
A base-word $\alpha$, which can be mapped to a codeword in this method,
must satisfy only the dependencies between $\alpha_i$
and $\alpha_{i+1}$ ($0\leq i \leq n-1$), but
it still can be non-realizable.
The $n$ cells are viewed cyclically, i.e.,
the charge of the last cell, $c_{n-1}$ is compared with
the charge in the first two cells, $c_0$ and $c_1$,
and the same works for the three charge levels $c_{n-2}$, $c_{n-1}$,
and $c_0$.
Therefore, there might exists a non-realizable dependency
between the charge levels in the last two cells
and the charge levels in the first two cells.
Such a non-realizable base-word will be called
a \emph{cyclic non-realizable} base-word.
For example, the following base-words are cyclic non-realizable.
\begin{itemize}
\item $1^n$ - the charge levels are always decreased.
\item $6^n$ - the charge levels are always increased.
\item $(2,5)^{n/2}$, where $n$ is even -
the charge level of each cell is between the charge levels
of the two cells which proceed it.
\end{itemize}
\begin{theorem}
\label{thm:deocodingIn1_3_nLRM}
Table~\ref{table:encodingKey13n} provides an one-to-one mapping
between the realizable base-words and the legal codewords
\end{theorem}
\begin{proof}
Obviously, each base-word is mapped to exactly one codeword.
Now, we prove that the other direction is also true.
Clearly, $1^n$ is an illegal codeword as
the charge level of each cell should be between the charge levels
of the two cells which proceed it.
Thus, given a legal codeword $g=(g_0,g_1,\ldots,g_{n-1})$,
there exists $0\leq i \leq n-1$, such that $g_i\in \{0,2\}$.
If $g_i=0$ then we have $\alpha_i\in \{1,3\}$, i.e., $\alpha_i$ is odd.
Thus, $\alpha_{i+1}$ is determined by an
entry in the first row in Table~\ref{table:encodingKey13n},
where the column is chosen by the value of $g_{i+1}$.
If $g_i=2$ then we have $\alpha_i\in \{4,6\}$, i.e., $\alpha_i$ is even.
Thus, $\alpha_{i+1}$ is determined by an entry
in the second row in Table~\ref{table:encodingKey13n},
where the column is chosen  by the value of $g_{i+1}$.
Now, it is easy to determine $\alpha_{i+2},\alpha_{i+3},\ldots, \alpha_{i+n-1}, \alpha_{i+n}=\alpha_i$
one after one in this cyclic order.
Note that if $\alpha_i$ is not equal
to an optional initial value
(from the set $\{1,3\}$ if $g_i=0$ and from $\{4,6\}$ if $g_i=2$)
then we can conclude that $g$ is illegal.
\end{proof}

Decoding a given codeword to a base-word
doesn't guarantee that the codeword is legal,
because the accepted base-word may be cyclic non-realizable.
For example, the cyclic non-realizable base-word
${\alpha=1^n}$ is mapped to the illegal codeword ${g=0^n}$.
Given such a codeword, it would be interesting to decide
efficiently if it is a legal codeword or not.

Next, the main theorem for $(1,3,n)$-LRM schemes is given.
\begin{theorem}
\label{thm:numberOfWordsIn1_3_nLRM}
If $M_3$ is the number of legal codewords in the $(1,3,n)$-LRM scheme
then $\lim\limits_{n\to \infty}\frac{M_3}{3^n}=1$.
\end{theorem}
\begin{proof}
Note that $g_i$ is determined by $c_i$, $c_{i+1}$,
and $c_{i+2}$.
$g_i=0$ if $c_{i+2}$ is lower than $c_i$ and $c_{i+1}$;
$g_i=1$ if $c_{i+2}$ is between $c_i$ and $c_{i+1}$;
and $g_i=2$ if $c_{i+2}$ is higher than $c_i$ and $c_{i+1}$.

Given a sub-codeword $(g_{0}, g_{1},\ldots, g_{i-2})$
obtained by the charge levels $c_0,c_1,\ldots,c_i$,
the relation between the charge levels in the $i$th cell, $c_i$,
and the first two cells, $c_0$ and $c_1$,
might have a few options.
These options will be denoted by $0$, $1$ and $2$,
where $0$ represents that $c_i$ is lower than $c_0$ and $c_1$,
$1$ represents that $c_i$ is between them,
and $2$ represents that $c_i$ is higher than both of them.
For each $i$, $3\leq i < n$,
we provide two properties regarding the charge levels $c_{i-1}$ and $c_{i}$:

\begin{itemize}
\item[(Q.1)] the permutation induced by $c_{i-1}$ and $c_{i}$ ($[1,2] $ or $[2,1]$);

\item[(Q.2)] the set of all possible pairs of the relations
between the charge levels $c_{i-1}$ and $c_i$ and the
charge levels $c_0$ and $c_1$.
\end{itemize}

The elements of the set defined in (Q.2) will be denoted by
pairs $(x,y)$, $x,y \in \{0,1,2\}$, where $x$ represents the relation between $c_{i-1}$
and the first two cells, and $y$ represents the relation between $c_{i}$
and the first two cells. Note, that not all the nine pairs $(x,y)$ can
be obtained for a given permutation defined by (Q.1).

We call each set of properties defined by (Q.1) and (Q.2) a \emph{state},
and the state at index $i$ will be denoted by~$P_i$.
Given a sub-codeword $g' = (g_0, g_1, \ldots, g_{n-4},g_{n-3})$,
the states in the sequence $(P_3,P_4,\ldots, P_{n-1})$ are
determined one after one in this order.
It is easy to verify that
if $P_i=P_j$ for some $3 \leq i <j < n-1$ then $P_{i+1}=P_{j+1}$.
A~state which has the all possibilities in the second
property will be called a \emph{complete state}.
It is easy to verify that in the $(1,3,n)$-LRM there are two
complete states.

\noindent
1) state~1 : $[1,2]$, $\{ (0,0),(1,1),(1,0),(2,2),(2,1),(2,0) \}$.
\noindent
2) state~2 : $[2,1]$, $\{ (0,0),(0,1),(0,2),(1,1),(1,2),(2,2) \}$.

Given $g_{i-1}$, the
succeeding state $P_{i+1}$ of a state $P_i$ which is a complete state,
is given in Table~\ref{table:fullStatesSeq13n}.
\begin{table}[!ht]
\begin{center}
\begin{tabular}{|l||l|l|l|}
\hline
\backslashbox{$P_i$}{$g_{i-1}$} & $0$		& $1$		& $2$		\\
\hline
\hline
	state~1						& state~1	& state~2	& state~2	\\
\hline
	state~2						& state~1	& state~1	& state~2	\\
\hline
\end{tabular}
\end{center}
\caption{Succeeding states in the $(1,3,n)$-LRM scheme}
\label{table:fullStatesSeq13n}
\end{table}

Denote by $\pi$ the permutation defined by the charge levels in the first two cells.
Given $\pi$ and $g' = (g_0, g_1, \ldots, g_{n-4},g_{n-3})$,
the sub-base-word $(\alpha_0,\alpha_1, \ldots, \alpha_{n-3})$
of a realizable base-word which corresponds to $\pi$ and $g'$ is determined unambiguously.
But, $\alpha_{n-2}, \alpha_{n-1}$ might have a few options.
These options are determined by the state $P_{n-1}$ and the permutation~$\pi$.
Each option provides a distinct base-word which is
represented by the state $P_{n-1}$ and the permutation~$\pi$.
For example the permutation $\pi=[2,1]$ and the
sub-codeword $g'=2^{n-2}$ imply that
the charge levels are always increased, where
$c_0$ is the lowest, and $c_{n-1}$ is the highest.
Therefore, the only base-word it represents
is $(\underbrace{6,6,\ldots, 6}_{n-2\ times},3,2)$,
where $P_{n-1}=([2,1],\{(2.2)\})$.

In Table~\ref{table:fullStatesNumWords13n}
we enumerate the number of base-words represented
by $\pi$ and $P_{n-1}$ for state~1 or state~2. Given~$g'$,
$c_0$ is compared with $c_{n-2}$ and $c_{n-1}$, to obtain $g_{n-2}$;
and $c_1$ is compared with $c_{n-1}$ and~$c_0$, to obtain $g_{n-1}$.
Note, that some different pairs $(x,y)$ of a given state result in the
same pair $(g_{n-2},g_{n-1})$ for a given $\pi$.
Note also that in Table~\ref{table:fullStatesNumWords13n},
the sum of values in each row and in each column equals to $9=3^2$.
\begin{table}[!ht]
\begin{center}
\begin{tabular}{|l||l|l|}
\hline
\backslashbox{$P_{n-1}$}{$\pi$} & $[1,2]$	& $[2,1]$	\\
\hline
\hline
	state~1	& 5									& 4									\\
\hline
	state~2	& 4									& 5									\\
\hline
\end{tabular}
\end{center}
\caption{The number of base-words represented by the complete
states in the $(1,3,n)$-LRM scheme}
\label{table:fullStatesNumWords13n}
\end{table}

If a sub-codeword $g'=(g_0, g_1, \ldots, g_{n-4}, g_{n-3})$
contains the sequence $(2,0,1,1)$ as a subsequence at indices $(i-3,i-2,i-1,i)$,
then this sequence ends with state~1, i.e. $P_{i+2}$ is state~1.
The reason is that in this case $c_{i+1}$ and $c_{i+2}$ have no dependency on the
charge levels of $c_{i-3}$ or $c_{i-2}$,
i.e., each one of $c_{i+1}$ and $c_{i+2}$ can be
lower than, between, or higher than $c_{i-3}$ and $c_{i-2}$.
Therefore, the relation between $c_{i+1}$ and $c_{i+2}$,
and the first two cells, $c_0$ and $c_1$,
has all the possibilities, i.e., $P_{i+2}$ is a complete state.
It is easy to verify that $c_{i+2}$ is lower than $c_{i+1}$,
and thus $P_{i+2}$ is state~1.
By Table~\ref{table:fullStatesSeq13n}
we have that $P_{n-1}$ of this sequence must
be a complete state (state~1 or state~2).

By using the well known Perron-Frobenius Theorem~\cite{Fro12,Per07},
we can compute the behavior of the number
of sequences of length $n-2$ over the alphabet $\{0,1,2\}$
which don't include $(2,0,1,1)$ as a subsequence.
First, an automata whose states accept
all the sequences over $\{0,1,2\}$
which don't contain the subsequence $(2,0,1,1)$ is given.
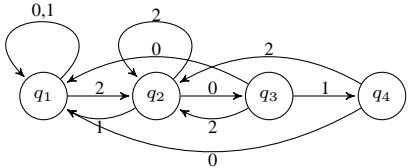
\begin{center}
\begin{tikzpicture}[->,>=stealth',shorten >=1pt,auto,node distance=1.5cm]
  \tikzstyle{every state}=[draw, minimum size=17pt, font={\scriptsize}]

  \node[state] (1)              {$q_1$};
  \node[state] (2) [right of=1] {$q_2$};
  \node[state] (3) [right of=2] {$q_3$};
  \node[state] (4) [right of=3] {$q_4$};

  \path
  		(1) edge[state]              node[yshift=-0.2cm]         {2}   (2)
  		(1) edge[state,loop]         node[above, yshift=-0.2cm]  {0,1} (1)
  		
  		(2) edge[state]              node[yshift=-0.2cm]         {0}   (3)
  		(2) edge[loop,state]         node[above, yshift=-0.2cm]  {2}   (2)
  		(2) edge[bend left, state]   node[yshift=0.2cm]          {1}   (1)
  		
  		(3) edge[state]              node[yshift=-0.2cm]         {1}   (4)
  		(3) edge[bend right,state]   node[above, yshift=-0.2cm]  {0}   (1)
  		(3) edge[bend left, state]   node[yshift=0.2cm]          {2}   (2)
  		
  		(4) edge[bend left, state]   node[yshift=0.2cm]          {0} (1)
  		(4) edge[bend right,state]   node[above, yshift=-0.2cm]  {2} (2);
\end{tikzpicture}
\end{center}

In the matrix which represents this automata (its state diagram),
the value in cell $(q_i, q_j)$
is the number of the labels on the outgoing edge from $q_i$ to $q_j$.
\begin{center}
\begin{tabular}{l||llll}
& $q_1$ & $q_2$ & $q_3$ & $q_4$\\
\hline
\hline
$q_1$ & $2$ & $1$ & $0$ & $0$ \\
$q_2$ & $1$ & $1$ & $1$ & $0$ \\
$q_3$ & $1$ & $1$ & $0$ & $1$ \\
$q_4$ & $1$ & $1$ & $0$ & $0$ \\
\end{tabular}
\end{center}
The largest real eigenvalue of this matrix is $2.9615$.
Therefore, by the Perron-Frobenius Theorem,
we can conclude that the number of sub-codewords
of length $n-2$ which don't contain $(2,0,1,1)$ as a subsequence,
tends to be $2.9615^{n-2}$, where $n$ tends to $\infty$.
Let $M'$ be the number of sub-codewords of length $n-2$
which contain $(2,0,1,1)$ as a subsequence.
By Table \ref{table:fullStatesNumWords13n} we can conclude that
the number of legal codewords in the $(1,3,n)$-LRM scheme
is $M_3 \geq 9M'$, and therefore
$\lim\limits_{n\to \infty}\frac{M_3}{3^n}=1$.
\end{proof}

\section{The $(1,t,n)$-LRM Scheme for $t \geq 4$}
\label{sec:1tnScheme}

Some of the results in Section~\ref{sec:13nScheme}
can be generalized
to the $(1,t,n)$-LRM Scheme, $t \geq 4$.
In the $(1,t,n)$-LRM scheme the size of each sliding window is~$t$.
Therefore, to present the demodulated sequences of permutations
the alphabet of the base-words has size $t!$.
Given $t$ consecutive charge levels, $c_i , c_{i+1},\ldots, c_{i+t-1}$,
the corresponding permutation related to $\alpha_i$, from $S_t$, is
uniquely determined by the order of the $t$
charge levels. The position of the symbol $t$ in this permutation
determines the value of $g_i$, i.e., $g_i = j$, $0 \leq j \leq t-1$,
if $t$ is in position $t-j$ in the permutation.
Therefore, a base-word uniquely determines the related codeword.
To obtain the original base-word from the given codeword would
be easy if for some $i$, $\alpha_i$ is given
(in fact only the permutation related to $t-1$ consecutive cells
is required). If no such permutation is given
then the task becomes more complicated.

Given a sub-codeword $(g_{0}, g_{1},\ldots, g_{i-t+1})$,
the relation between the charge levels in the $i$th cell, $c_i$,
and the first $t-1$ cells, $c_0,c_1, \ldots , c_{t-2}$,
might have a few options.
These options will be denoted by $0,1,\ldots,t-1$,
where $j$, $0 \leq j \leq t-1$, represents that $c_i$
is in position $t-j$ among the $t$ charge levels
$c_0, c_1 , \ldots , c_{t-2}$ and $c_i$ (counting from the highest,
the first position, to the lowest, the $t$th position).
For each $i$, $2t-3 \leq i \leq n-1$,
we provide two properties regarding the charge
levels $c_{i-t+2},c_{i-t+3},\ldots,c_i$:

\begin{itemize}
\item[(Q.1)] the permutation induced by $c_{i-t+2},c_{i-t+3},\ldots,c_i$;

\item[(Q.2)] the set of all possible $(t-1)$-tuples of the
relations between the charge levels
$c_{i-t+2},c_{i-t+3},\ldots,c_i$ and the charge levels
of the first $t-1$ cells.
\end{itemize}

The elements of the set defined in (Q.2) will be denoted by
a $(t-1)$-tuple $(x_1, x_2, \ldots , x_{t-1})$, $x_j \in \{0,1,\ldots , t-1\}$,
$1 \leq j \leq t-1$, where $x_j$ represents the relation between $c_{i-t+j+1}$
and the first $t-1$ cells. Note, that not all the $t^{t-1}$ $(t-1)$-tuples can
be obtained for a given permutation defined by (Q.1).

We call each set of properties defined by (Q.1) and (Q.2) a \emph{state},
and the state at index $i$ will be denoted by $P_i$.
Given a sub-codeword $g' = (g_0, g_1, \ldots, g_{n-t-1},g_{n-t})$,
the states in the sequence $(P_{2t-3},P_{2t-2},\ldots, P_{n-1})$ are
determined one after one in this order.
It is easy to verify that
if $P_i=P_j$ for some $2t-3 \leq i < j < n-1$ then $P_{i+1}=P_{j+1}$.
A state which has the all possibilities in the
second property will be called a \emph{complete state}.
It is easy to verify that in the $(1,t,n)$-LRM there
are $(t-1)!$ complete states (defined by the permutations
of (Q.1)~).

Given $P_i$ which is a complete state and any value of $g_{i-1}$,
it is easily verified that
the succeeding state $P_{i+1}$ is also a complete state,

Denote by $\pi$ the permutation defined by the
charge levels in the first $t-1$ cells. Given $\pi$ and
$g' = (g_0, g_1, \ldots, g_{n-t-1},g_{n-t})$,
the sub-base-word $(\alpha_0,\alpha_1, \ldots, \alpha_{n-t})$
of a realizable base-word which corresponds
to $\pi$ and $g'$ is determined unambiguously.
But, $\alpha_{n-t+1}, \alpha_{n-t+2}, \ldots , \alpha_{n-1}$ might have a few options.
These options are determined by the state $P_{n-1}$ and the permutation~$\pi$.
Each option provides a distinct base-word which is
represented by the state $P_{n-1}$ and the permutation~$\pi$.

We generate a table to
enumerate the number of base-words represented
by $\pi$ and $P_{n-1}$ for the $(t-1)!$ complete states.
Given $g'$,
$c_0$ is compared with $c_{n-t+1},c_{n-t+2},\ldots,c_{n-1}$, to obtain $g_{n-t+1}$;
$c_1$ is compared with $c_{n-t+2},\ldots,c_{n-1},c_0$, to obtain $g_{n-t+2}$,
and so on until
$c_{t-2}$ is compared with $c_{n-1},c_0,\ldots,c_{t-3}$, to obtain $g_{n-1}$.
It can be proved that in this table
the sum of values in each row and in each column is equal to~$t^{t-1}$.

The next step is to find a sequence $(a_1, a_2 , \ldots , a_r)$,
$r \geq t$, with the following properties.
If a sub-codeword $g'=(g_0, g_1, \ldots, g_{n-t})$
contains the sequence $(a_1, a_2 , \ldots , a_r)$
as a subsequence at indices $(i-r+1,i-r+2, \ldots ,i-1,i)$,
then this sequence ends in a complete state, i.e. $P_{i+t-1}$ is one
of the $(t-1)!$ complete states.
The reason is that in this case $c_{i+1},c_{i+2},\ldots,c_{i+t-1}$
have no dependency on the
charge levels of $c_{i-r+1},c_{i-r+2},\ldots,c_{i-r+t-1}$,
i.e., each one of the charge levels $c_{i+1},c_{i+2},\ldots,c_{i+t-1}$
can be lower than $c_{i-r+1},c_{i-r+2},\ldots,c_{i-r+t-1}$,
between them ($t-2$ options), or higher than all of them.
Therefore, the relations between $c_{i+1},c_{i+2},\ldots,c_{i+t-1}$,
and the charge levels in the
first $t-1$ cells, $c_0, c_1 ,\ldots , c_{t-2}$,
have all the possibilities, i.e., $P_{i+t-1}$ is a complete state.
This implies that also $P_{n-1}$ is a complete state.

Now, the Perron-Frobenius Theorem can be used
to compute the number of sequences, of length
$n-t+1$ over alphabet $\{0,1,\ldots ,t-1\}$,
which don't include $(a_1,a_2, \ldots , a_r)$ as a subsequence.
It is required that this number will tend to $\beta^{n-t+1}$
where $n$ tends to $\infty$ and
$\beta$ is a constant (related to the largest eigenvalue
of the transition matrix of the related automata) for which
${\beta < t}$. Now, it can be concluded with the generated table
that if $M_t$ is the number of legal codewords in the $(1,t,n)\text{-}$LRM scheme
then ${\lim\limits_{n\to \infty}\frac{M_t}{t^n}=1}$.

For $t=4$ one required such subsequence is
$(3,3,0,1,2,1)$ for which $\beta = 3.99902$, and hence
\begin{theorem}
\label{thm:numberOfWordsIn1_4_nLRM}
$\lim\limits_{n\to \infty}\frac{M_4}{4^n}=1$.
\end{theorem}

The problem of finding the value of $\lim\limits_{n\to \infty}\frac{M_t}{t^n}$
can be formulated as a purely combinatorial problem.
Let
$$
\begin{array}{c}
\cA_t^n \deff \{ (a_0 , \ldots , a_{t^n-1}) ~:~ a_j \in \Z , ~ 0 \leq j \leq t^n -1, \\
| \{ a_j , a_{j+1}, \ldots , a_{j+t-1} \} | =t \},
\end{array}
$$
$\Pi (b_1 , b_2 ,\ldots , b_t )$ be the permutation from $S_t$ defined by the
subsequence $(b_1 , b_2 ,\ldots , b_t )$, and
$$
\begin{array}{c}
\cS_t^n \deff \{ (\pi_0 , \ldots , \pi_{t^n-1}) ~:~ \pi_j \in S_t , ~ 0 \leq j \leq t^n -1, \\
 \pi_j = \Pi ( a_j , a_{j+1}, \ldots , a_{j+t-1} ), ~ (a_0 , \ldots , a_{t^n-1}) \in \cA_t^n  \},
\end{array}
$$
where the indices are taken modulo $t^n$.

What is the value of $\lim\limits_{n\to \infty}\frac{|\cS_t^n|}{t^n}$?
We conjecture that the value is 1 and proved this value for $t=3$ and $t=4$.

\section{Constant Weight $(1,2,n)$-LRM Gray Codes}
\label{sec:weight}

One important topic related to rank modulation is
an order of the codewords in such a way that each
codeword will define an alphabet letter. This implies that
$n$ consecutive cells define an alphabet letter
and any change in the charge levels of some cells
relates to a change in the alphabet letter.
The most effective ordering is a Gray code ordering,
i.e., a codeword is obtained by a minimal change
in the codeword which proceed it. This should be
a consequence of a small change in the related charge levels.
In the rank modulation scheme this change in the charge levels
is obtained by the push-to-the-top operation in a window of length $t$.
We will concentrate only in the case where $t=2$ since in this case
the windows have length two and hence the permutations
are from $S_2$, i.e., we can use binary codewords.
In this respect, we will be interested also in the case
where all the codewords have the same weight.
We define an $(1,2,n;w)$-LRMGC to be an $(1,2,n)$-LRM Gray code,
where the codewords are ordered in an order which defines
a Gray code and each codeword has weight~$w$.
If all the codewords have the same weight then
we can bound the difference between the charge levels of
a cell on which the push-to-the-top operation is performed~\cite{GLSB11}.
Two codewords are adjacent only if they differ in two positions, where 01
can be changed to 10. In this respect, the last codeword and
the first codeword are also considered to be adjacent. It was conjectured
in~\cite{GLSB11} that an $(1,2,n;2)$-LRMGC has at most $2n$ codewords
and such a code with~$2n$ codewords was constructed.
We have been able to prove this conjecture, i.e.,

\begin{theorem}
An $(1,2,n;2)$-LRMGC has at most~$2n$ codewords.
\end{theorem}

\section{Conclusions and Open Problems}
\label{sec:conclusion}
In this paper, encoding, decoding, and enumeration of the
$(1,t,n)$-LRM scheme are studied. A complete solution
is given for the $(1,3,n)$-LRM scheme.
Encoding for the $(1,t,n)$-LRM scheme for each $t\geq 3$ is presented.
For $(1,3,n)$ a related decoding was presented.
We also proved for $t\in\{3,4\}$ that if $M_t$ is the number
of legal codewords in the $(1,t,n)$-LRM scheme then
${\lim\limits_{n\to \infty}\frac{M_t}{t^n}=1}$.
We conclude with some problems for future research
raised in our discussion.
\begin{itemize}
\item Find an efficient algorithm to determine if a given
codeword in the $(1,t,n)$-LRM scheme, for $t \geq 3$, is
legal or not.

\item Find an efficient decoding algorithm
for the $(1,t,n)$-LRM scheme, $t \geq 4$.

\item Prove that if $t >4$ then
${\lim\limits_{n\to \infty}\frac{M_t}{t^n}=1}$.

\item For $w >2$, find optimal $(1,2,n;w)$-LRMGC.

\end{itemize}

\section*{Acknowledgment}
This work was supported in part by the U.S.-Israel
Binational Science Foundation, Jerusalem, Israel, under
Grant No. 2012016.


\end{document}